%% file: main.tex
\documentclass[aps,prl,reprint,nofootinbib,superscriptaddress]{revtex4-2}
\pdfoutput=1

\input{settings.tex}

\begin{document}
\input{preamble.tex}

\input{sections/introduction.tex}
\input{sections/graph_games.tex}
\input{sections/examples.tex}

\input{sections/discussion.tex}
\input{sections/acknowledgements.tex}

\bibliographystyle{apsrev4-2}
\bibliography{biblio}

\input{sections/supplementary.tex}

\end{document}

%% file: settings.tex
\usepackage[utf8]{inputenc}        
\usepackage{lmodern}               
\usepackage[T1]{fontenc}           
\usepackage[english]{babel}

\newif\ifdebug
\debugfalse
\usepackage{latexsym}              
\usepackage{amsmath}               
\usepackage{amssymb}               
\usepackage{amsfonts}              
\usepackage{amsthm}                
\allowdisplaybreaks

\usepackage{mathtools}             
\usepackage{nicefrac}

\usepackage{amsbsy}                 
\usepackage{bbm}                    

\usepackage{color, xcolor}          
\definecolor{darkgray}{rgb}{0.15,0.15,0.15}
\definecolor{lightgray}{rgb}{0.94,0.94,0.94}
\definecolor{lightlightgray}{rgb}{0.97,0.97,0.97}
\definecolor{darkred}{rgb}{0.80,0.00,0.00}
\definecolor{darkgreen}{rgb}{0.00,0.70,0.00}
\definecolor{darkblue}{rgb}{0.00,0.00,0.70}

\usepackage{tikz}                  
\usepackage{tikzsymbols}
\usetikzlibrary{arrows,calc,positioning,shapes}

\usepackage{pgfplots}
\pgfplotsset{compat=1.18}

\usepackage{enumerate}             
\usepackage{enumitem}
\setlist{nolistsep}

\usepackage{xparse}                
\usepackage{hyperref}   
\usepackage{url}
\usepackage{cleveref}
\hypersetup{
    pdfpagemode=UseOutlines,
    linkbordercolor = [rgb]{1, 1, 1},
    filebordercolor = [rgb]{1, 1, 1},
    citebordercolor = [rgb]{1, 1, 1},
    menubordercolor = [rgb]{1, 1, 1},
    urlbordercolor = [rgb]{1, 1, 1},
    colorlinks = true,
    linkcolor = [rgb]{0, 0.2, 0.5},
    citecolor = [rgb]{0, 0.35, 0},
    urlcolor = [rgb]{0.6, 0.2, 0.4},
    breaklinks = true
}

\usepackage{utfsym}

\usepackage{ragged2e}
\usepackage{graphicx}              
\graphicspath{{figures/}}          

\usepackage{float}                 
\usepackage{framed}                
\usepackage[framemethod=tikz]{mdframed}
\usepackage{enumitem}              
\usepackage{glossaries}            

\theoremstyle{plain}  
\newtheorem{thm}{Theorem}
\newtheorem{lem}[thm]{Lemma}

\newtheorem{cor}[thm]{Corollary}

\theoremstyle{definition}
\newtheorem{defn}[thm]{Definition}

\theoremstyle{remark}  

\theoremstyle{plain}
\newtheorem{innercustomgeneric}{\customgenericname}
\providecommand{\customgenericname}{}
\newcommand{\newcustomtheorem}[2]{%
  \newenvironment{#1}[1]
  {%
   \renewcommand\customgenericname{#2}%
   \renewcommand\theinnercustomgeneric{##1}%
   \innercustomgeneric
  }
  {\endinnercustomgeneric}
}

\newcustomtheorem{customthm}{Theorem}
\newcustomtheorem{customlem}{Lemma}
\newcustomtheorem{customcor}{Corollary}

\providecommand*{\hr}[1][class-arg]{%
    \hspace*{\fill}\hrulefill\hspace*{\fill}
    \vskip 0.65\baselineskip
}

\newcommand{\brac}[1]{\left(#1\right)}

\newcommand{\abs}[1]{\left| #1 \right|}

\NewDocumentCommand{\prob}{mg}{
    \IfValueTF{#2}{
        P(#1 \,\vert\, #2)
    }{
        P(#1)
    }
}

\newcommand\restr[2]{{
  \left.\kern-\nulldelimiterspace 
  #1 
  \right|_{#2} 
  }}

\let\latexchi\chi
\makeatletter
\renewcommand\chi{\@ifnextchar_\sub@chi\latexchi}
\newcommand{\sub@chi}[2]{
  \@ifnextchar^{\subsup@chi{#2}}{\latexchi^{}_{#2}}%
}
\newcommand{\subsup@chi}[3]{
  \latexchi_{#1}^{#3}%
}
\makeatother

\newcommand{\bra}[1]{\langle#1|}
\newcommand{\ket}[1]{|#1\rangle}

\makeatletter
\newcommand\RedeclareMathOperator{%
  \@ifstar{\def\rmo@s{m}\rmo@redeclare}{\def\rmo@s{o}\rmo@redeclare}%
}
\newcommand\rmo@redeclare[2]{%
  \begingroup \escapechar\m@ne\xdef\@gtempa{{\string#1}}\endgroup
  \expandafter\@ifundefined\@gtempa
     {\@latex@error{\noexpand#1undefined}\@ehc}%
     \relax
  \expandafter\rmo@declmathop\rmo@s{#1}{#2}}
\newcommand\rmo@declmathop[3]{%
  \DeclareRobustCommand{#2}{\qopname\newmcodes@#1{#3}}%
}
\@onlypreamble\RedeclareMathOperator
\makeatother

\RedeclareMathOperator{\Re}{Re}

\colorlet{figureColorPalette1}{olive}
\colorlet{figureColorPalette2}{teal}
\colorlet{figureColorPalette3}{orange}
\colorlet{figureColorPalette4}{magenta}
\colorlet{figureColorPalette5}{darkgreen}
\tikzset{blackVertex/.style = {circle,
                         draw = white,
                         line width = 1.5pt,
                         fill = black!80!white,
                         inner sep = 3pt},
        redVertex/.style = {circle,
                         draw = white,
                         line width = 1.25pt,
                         fill = red!80!white,
                         inner sep = 2.5pt},
        labeledWhiteVertex/.style = {circle,
                         fill = white,
                         inner sep = 6pt},
        labeledBlackVertex/.style = {circle,
                         draw = black!80!white,
                         line width = 1.5pt,
                         inner sep = 2pt,
                         font = \tiny},
        edge/.style = {line width = 2pt},
        whiteEdge/.style = {-,
                            line width = 3pt,
                            draw = white},
        blackEdge/.style = {-,
                            line width = 2pt,
                            draw = black},
        redEdge/.style = {-,
                            line width = 1.5pt,
                            draw = red},
        pair/.style = {ellipse,
                       fill,
                       opacity=0.25,
                       minimum height=1.5em},
        labeledPair/.style = {ellipse,
                       fill,
                       opacity=0.25,
                       minimum height=2em}}

%% file: preamble.tex
\title{Monogamy of Nonlocal Games}

\author{David Cui}
\affiliation{Department of Mathematics, Massachusetts Institute of Technology, Cambridge, Massachusetts 02139, USA.}
\author{Arthur Mehta}
\author{Denis Rochette}
\affiliation{Department of Mathematics and Statistics, University of Ottawa, Ottawa, Ontario K1H 8M5, Canada.}

\begin{abstract}
    Bell monogamy relations characterize the trade-offs in Bell inequality violations among pairs of players in multiplayer settings. In this work, we introduce a method for extending monogamy relations from a distinguished set of configurations to monogamy relations on all possible multiplayer settings. Applying this approach, we show that nonlocality in the CHSH game arises in only two cases: the original two-player scenario and the four-player scenario on a line. While the bound for this four-player scenario follows from known quadratic monogamy constraints, we also establish two new six-party numerical monogamy relations that cannot be derived from existing results. In particular, we show there are points in the intersection of consecutive quadratic Bell monogamy relations which are not quantum realizable. Finally, we present a nonlocal game in which a single player can simultaneously saturate the quantum value with two other parties. This is the first known example of a nonlocal game unaffected by the monogamous nature of quantum entanglement. 
\end{abstract}

\maketitle

%% file: sections/introduction.tex
\section{Introduction}

Entanglement is one of the defining features of quantum mechanics. This resource can lead to correlations that no local hidden-variable theory can reproduce. The phenomenon associated with the existence of such correlations is known as \emph{nonlocality}. On the other hand, entanglement cannot be freely shared among multiple parties. This fundamental limitation, established by Coffman, Kundu, and Wootters~\cite{Coffman2000} and known as the \emph{monogamy of entanglement (MOE)} \cite{Koashi2004, Osborne2006, Regula2014}, places constraints on how entanglement is distributed in multipartite scenarios. Understanding how MOE interacts with nonlocality is important both from a foundational perspective \cite{Streltsov2012,Lancien2016} and for a range of applications \cite{Pawowski2010, Augusiak2014}.

In 1964, Bell showed that quantum mechanics is incompatible with any local hidden-variable theory \cite{bell}. A few years later, Clauser, Horne, Shimony, and Holt~\cite{chsh} introduced a two-party test---now known as the CHSH nonlocal game---that demonstrates outcome statistics impossible to explain with classical means. Since then, nonlocality in the two-party setting has been extensively studied using various nonlocal games \cite{Mermin, Peres, Aravind}. MOE can also be investigated using nonlocal games. Toner~\cite{toner2009monogamy} considered a three-party extension of the CHSH game with Alice, Bob, and Charlie, where the test randomly involves either Alice and Bob or Bob and Charlie. In this configuration, MOE implies that entanglement provides no advantage over classical correlations. Toner and Verstraete~\cite{0611001} then extended these results, characterizing the trade-off in nonlocality between the Alice--Bob and Bob--Charlie pairs.

The three-player scenario described above serves as a prototypical \emph{Bell monogamy relation}, quantifying how MOE restricts simultaneous violations of multiple Bell inequalities. However, its optimal quantum winning probability is $\frac{3}{4}$, achievable without any use of entanglement \cite{toner2009monogamy}. This bound gives rise to the following Bell monogamy relation: $\mathcal{B}_{AB} + \mathcal{B}_{BC} \leq 4$.

In this work, we formalize these multipartite scenarios by extending symmetric $2$-player games $\mathcal{G}$ to $n$-player games $\mathcal{G}^H$, where players occupy the vertices of a graph $H=(V,E)$ (see \cref{fig:game}). A referee selects an edge uniformly at random and plays $\mathcal{G}$ with the players on the vertices. The resulting Bell monogamy relation is given by the average of the two party Bell violations across all edges. 

We introduce a technique which allows one to extend a Bell monogamy relations on smaller fundamental graphs, to obtain new Bell monogamy relations on all possible configurations with any number of players. In particular this allows us to give give general sufficient conditions under which no quantum strategy can outperform the optimal classical strategy for any graph with more than two players. We call a $2$-player game satisfying this property \emph{monogamous}.

We then apply these criteria to a range of well-known two-player games, fully determining their behavior on all graphs $H$. Surprisingly, for the CHSH game, average case nonlocality appears only in two scenarios: the original $2$-player setting and the $4$-player setting played on the path graph $P_{4}$ (see \cref{fig:P3andP4}). It is known for the for the cases of $P_3$ and $P_4$ the set of quantum realizable correlations are determined exactly by studying the intersection of consecutive quadratic Bell monogamy relations, $\mathcal{B}_{XY}^2 + \mathcal{B}^2_{YZ} \leq 8$. Numerically, we confirm that there is no quantum advantage on two $6$ party arrangements but that studying consecutive quadratic Bell monogamy relations is not strong enough to determine this. This shows that taking the intersection of the quadratic Bell monogamy relations does not exactly characterize the set of quantum correlations on $6$ parties as it does on $P_{3}$ and $P_{4}$ \cite{0611001, ramanathan2018bell}.

\begin{figure}[ht]
    \centering
    \begin{tikzpicture}
        \coordinate (0) at (0,0) {};
        \coordinate (1) at (81:3.5em) {};
        \coordinate (2) at (130:4.2em) {};
        \coordinate (3) at (180:3.7em) {};
        \coordinate (4) at (175:7.5em) {};
        \coordinate (5) at (6:5.5em) {};

        \path let 
            \p1 = ($(5)-(0)$),
            \n1 = {veclen(\x1,\y1) + 2em/1pt},
            \n2 = {acos(\x1 / (\n1 - 2em/1pt))}
            in node[labeledPair, color=figureColorPalette1, minimum width=\n1, rotate around={\n2:($(0.center)!0.5!(5.center)$)}, label={[anchor=north, rotate around={\n2:($(0.center)!0.5!(5.center)$)}, label distance=-2pt]south:{\small$\omega_{AF}$}}] at ($(0.center)!0.5!(5.center)$) {};
            
        \foreach \i/\j in {0/1, 0/2, 0/3, 0/5, 3/4, 5/1} {
            \draw[whiteEdge] (\i.center) -- (\j.center);
            \draw[blackEdge] (\i.center) -- (\j.center);
        }

        \foreach \i/\n in {0/A, 1/B, 2/C, 3/D, 4/E, 5/F} {
            \node[labeledWhiteVertex] at (\i.center) {};
            \node[labeledBlackVertex] at (\i.center) {$\n$};
        }
        
    \end{tikzpicture}
    \caption{Several players who share a global state are represented as the vertices of a graph $H$. A pair of players are chosen to play a symmetric $2$-player game $\mathcal{G}$, by sampling an edge from $H$. The edge $(A,F)$ is chosen and the pair of player win the game $\mathcal{G}$ with probability $\omega_{AF}$.}
    \label{fig:game}
\end{figure}

In contrast, we establish that other well-studied games, such as the Odd Cycle game \cite{Cleve}, are monogamous. Finally, we give an example of a $2$-player nonlocal game for which it is not possible to obtain a Bell monogamy relation for the three party configuration on a line. In particular Bob can saturate the quantum value with both Alice and Charlie simultaneously. This is the first such example of a non-monogamous Bell inequality, which had been believed to be impossible.

We note that our techniques are applicable in a broader context. In particular, they can be used to determine when no advantage exists over classical strategies for commuting-operator correlations \cite{Tsirelson1993SomeRA,Vid19} or even more general non-signaling correlations \cite{Popescu1994}. The latter are especially important for understanding the boundary between physically realizable and non-physical theories \cite{Pawowski2009,Ramanathan2014,Botteron2024}.

In previous work, Bell monogamy relations have been studied in the context of correlation complementarity, where $n$ observers each perform two measurements yielding dichotomic $\pm 1$ outcomes \cite{Kurzyski2011, PhysRevLett.88.210401}. Tran et al.~\cite{ramanathan2018bell} investigated Bell monogamy relations for observers arranged in a network, as illustrated in \cref{fig:game}, and derived a family of tight monogamy relations for an arbitrary number of observers. Unlike our approach, their method selects edges of the graph $H$ according to some non-uniform distribution determined by the line graph $L(H)$. These monogamy relations follow directly from combining a three-party Bell monogamy relation with the handshaking lemma from graph theory. By contrast, the monogamy relations we establish cannot, in general, be obtained solely from such three-party relations.

Of particular relevance is Tran et al.’s analysis of general Bell inequalities on $P_{4}$, which characterizes the quantum trade-offs imposed by three-party Bell monogamy relations. Our optimal $P_{4}$ strategy for CHSH falls within this characterization.

\begin{figure}[ht]
    \centering
    \hfill
        \begin{tikzpicture}
            \coordinate (1) at (0em,0) {};
            \coordinate (2) at (4em,0) {};
            \coordinate (3) at (8em,0) {};

            \node at (4em,-4em) {$\omega^*\brac{\mathrm{CHSH}^{P_3}} = 0.75$};
    
            \path let \p1 = ($(1)-(2)$), \n1 = {1em + veclen(\x1,\y1)} in node[pair, color=figureColorPalette1, minimum width=\n1, label={[anchor=north]south:{\scriptsize$\omega^*\!=\!0.75$}}] at ($(1.center)!0.5!(2.center)$) {};
            \path let \p1 = ($(2)-(3)$), \n1 = {1em + veclen(\x1,\y1)} in node[pair, color=figureColorPalette2, minimum width=\n1, label={[anchor=north]south:{\scriptsize$\omega^*\!=\!0.75$}}] at ($(2.center)!0.5!(3.center)$) {};
            
            \foreach \i/\j in {1/2, 2/3} {
                \draw[whiteEdge] (\i.center) -- (\j.center);
                \draw[blackEdge] (\i.center) -- (\j.center);
            }
    
            \foreach \i in {1,...,3}
                \node[blackVertex] at (\i.center) {};
    
        \end{tikzpicture}
    \hfill
        \begin{tikzpicture}
            \coordinate (1) at (0em,0) {};
            \coordinate (2) at (4em,0) {};
            \coordinate (3) at (8em,0) {};
            \coordinate (4) at (12em,0) {};

            \node at (6em,-4em) {$\omega^*\brac{\mathrm{CHSH}^{P_4}} \approx 0.76$};
    
            \path let \p1 = ($(1)-(2)$), \n1 = {1em + veclen(\x1,\y1)} in node[pair, color=figureColorPalette1, minimum width=\n1, label={[anchor=north]south:{\scriptsize$\omega^*\!\approx\!0.81$}}] at ($(1.center)!0.5!(2.center)$) {};
            \path let \p1 = ($(2)-(3)$), \n1 = {1em + veclen(\x1,\y1)} in node[pair, color=figureColorPalette2, minimum width=\n1, label={[anchor=north]south:{\scriptsize$\omega^*\!\approx\!0.65$}}] at ($(2.center)!0.5!(3.center)$) {};
            \path let \p1 = ($(3)-(4)$), \n1 = {1em + veclen(\x1,\y1)} in node[pair, color=figureColorPalette3, minimum width=\n1, label={[anchor=north]south:{\scriptsize$\omega^*\!\approx\!0.81$}}] at ($(3.center)!0.5!(4.center)$) {};
            
            \foreach \i/\j in {1/2, 2/3, 3/4} {
                \draw[whiteEdge] (\i.center) -- (\j.center);
                \draw[blackEdge] (\i.center) -- (\j.center);
            }
    
            \foreach \i in {1,...,4}
                \node[blackVertex] at (\i.center) {};
    
        \end{tikzpicture}
    \hfill{}
    \caption{Optimal quantum winning probabilities of the $\mathrm{CHSH}$ game on the path graphs $P_3$ (left) and $P_4$ (right).}
    \label{fig:P3andP4} 
\end{figure}

%% file: sections/graph_games.tex
\section{Nonlocal Games over Graphs}

In a $2$-player \emph{symmetric nonlocal game} $\mathcal{G}$ (or just \emph{game} in the letter), the players Alice and Bob are given questions $x_1,x_2 \in \mathcal{I}$, sampled independently according to a uniform distribution, and provide answers $a_1,a_2 \in \mathcal{O}$ according to a quantum correlation $p(a_1,a_2|x_1,x_2)$. The winning condition is determined by a referee predicate $v(x_1,x_2,a_1,a_2)$ that satisfies the symmetric condition: $v(x_1,x_2,a_1,a_2) = v(x_2,x_1,a_2,a_1)$. We write $\omega(\mathcal{G}; S)$ for the winning probabilities of a game $\mathcal{G}$ given a strategy $S$. The maximal classical and quantum winning probabilities are denoted  $\omega(\mathcal{G})$ and $\omega^*(\mathcal{G})$ respectively. We say that the game $\mathcal{G}$ has a \emph{quantum advantage} if $\omega(\mathcal{G}) < \omega^*(\mathcal{G})$.

In this letter, a \emph{graph} $H=(V,E)$ is a no multiple edge undirected graph, where $V$, $E$ are finite vertex and edge sets, respectively. We denote the path graph on $n$ vertices by $P_{n}$. The family of graphs $\mathcal{T}_k$ denotes the collection of trees of size $2k-1$ formed by joining $k$ copies of $P_{2}$ using $k-1$ edges (as illustrated in \cref{fig:T3andT4}).

A \emph{graph homomorphism} between a graphs $H_1$ a graph $H_1$ is a function that maps the vertices of $H_1$ to the vertices of $H_2$ that preserves adjacency: adjacent vertices in $H_1$ are mapped to adjacent vertices in $H_2$. We write $H_1 \to H_2$ if there exists a homomorphism from $H_1$ to $H_2$. For example, bipartite graphs are precisely those that admit a graph homomorphism to $P_2$. Moreover, if the strategy graph contains a loop, then any graph admits a graph homomorphism to it.

\begin{defn}
    The \emph{strategy graph} $\mathcal{S}_{\mathcal{G}}=(V,E)$ of a game $\mathcal{G}$ is a (not necessarily simple or connected) graph, with vertices consisting of functions $f:\mathcal{I} \rightarrow \mathcal{O}$, and with edges consisting of the pairs of functions $(f_A, f_B)$ that achieve the game’s optimal classical winning probability $\omega(\mathcal{G})$ (see for example \cref{fig:strategy_graph}).
\end{defn}

\begin{figure}[ht]
    \centering
    \begin{tikzpicture}
            \coordinate (1) at (0,0) {};
            \coordinate (2) at (3em,0) {};
            \coordinate (3) at (6em,0) {};
            \coordinate (4) at (9em,0) {};
    
            \foreach \i/\j in {1/2, 2/3, 3/4} {
                \draw[whiteEdge] (\i.center) -- (\j.center);
                \draw[blackEdge] (\i.center) -- (\j.center);    
            }

            \draw[whiteEdge, every loop/.style={min distance=3em}] (1.center) edge[in=130,out=-130,loop] ();
            \draw[blackEdge, every loop/.style={min distance=3em}] (1.center) edge[in=130,out=-130,loop] ();

            \draw[whiteEdge, every loop/.style={min distance=3em}] (4.center) edge[in=50,out=-50,loop] ();
            \draw[blackEdge, every loop/.style={min distance=3em}] (4.center) edge[in=50,out=-50,loop] ();
    
        \foreach \i in {1,...,4}
            \node[blackVertex] at (\i.center) {};
    \end{tikzpicture}
    \caption{The strategy graph of the CHSH game. The vertices on the extreme left and right are the constant strategies, while the two intermediate vertices correspond to the identity and flip strategies.}
    \label{fig:strategy_graph}
\end{figure}

\begin{defn}
    Given a simple connected graph $H = (V,E)$, with $|V|=n$ and a game $\mathcal{G}$, with referee predicate $v(x_1,x_2,a_1,a_2)$, we define an $n$-player game $\mathcal{G}^{H}$ with referee predicate
    \begin{equation*}
        v^{H}(x_{1}, \dots, x_{n}, a_{1}, \dots, a_{n}) \coloneqq \frac{1}{\abs{E}} \sum_{(i,j) \in E} v(x_{i}, x_{j}, a_{i}, a_{j}).
    \end{equation*}
\end{defn}
A straightforward computation shows that the winning probability of $\mathcal{G}^{H}$ can be expressed as a uniformly weighted sum over the winning probability of $\mathcal{G}$ on edge $(i,j) \in E$ as illustrated in \cref{fig:game}. A game $\mathcal{G}$ is said to have \emph{quantum advantage on a graph} $H$ if $\omega(\mathcal{G}^H) < \omega^*(\mathcal{G}^H)$.

Our main theorem gives a method to determine for which graphs $H$ the quantum value collapses to the classical 2-player value, $\omega(\mathcal{G}) =\omega(\mathcal{G}^H) =\omega^*(\mathcal{G}^H)$.

It should be noted that the classical value of a game is not generally the same in all graphs. For example, consider a 2-player binary output game with referee predicate $v(x_1, x_2, a_1, a_2) = 1 - \delta_{a_1 a_2}$ (i.e., the outputs of the two players are require to be different). Then the optimal classical value is $1$, but it is not possible for 3 players in the $C_3$ configuration to obtain this value. In contrast, for games where the strategy graph contains a loop, the classical value of a game is the same in all graphs.

Below we give a characterization of when the classical value does not decrease on a graph $H$.

\begin{lem} \label{lem:classical_graph_strategy}
    The classical value of a game $\mathcal{G}$ on a graph $H$ is the same as the classical value of the game if and only if there is a graph homomorphism $H \to \mathcal{S}_{\mathcal{G}}$.
\end{lem}
\begin{proof}    
    If the classical value of $\mathcal{G}^H$ equals the classical value of $\mathcal{G}$, then any optimal classical deterministic strategy for $\mathcal{G}^H$ induces a graph homomorphism from $H$ to $H^{\mathcal{G}}$. Specifically, for each vertex of $H$ is assigned a classical deterministic strategy of $\mathcal{G}$. For each edge in $H$, the paired strategies achieve exactly the winning probability $\omega(\mathcal{G})$; if this were not the case, we would have $\omega(\mathcal{G}) < \omega(\mathcal{G}^H)$, contradicting our assumption. Consequently, this pair forms an edge in $H^{\mathcal{G}}$. Since any optimal classical strategy can be taken to be deterministic by convexity, this completes the first direction of the proof.
    
    For the converse, suppose there exists a graph homomorphism from $H$ to the strategy graph $H^{\mathcal{G}}$. This graph homomorphism defines an optimal classical deterministic strategy for $\mathcal{G}^H$ by assigning a deterministic strategy of $\mathcal{G}$ to each vertex in $H$ and ensuring that for any edge in $H$, the winning probability is precisely $\omega(\mathcal{G})$.
\end{proof}

We then give a technique for establishing quantum upperbounds for families of graphs. Before stating this theorem, we provide an illustrative example of how, for certain graphs $H$, the quantum value $\omega^*(\mathcal{G}^{H})$ can be upperbounded by the quantum value on $P_3$. In \cref{fig:fractionalP3Decomposition} we consider playing the game on the $3$-cycle graph, $C_3$, with vertices $A$, $B$ and $C$. One can uniformly cover $C_3$ with $3$ weighted copies, $P_{ABC}, P_{BCA}, P_{CAB}$ of the graph $P_3$. Since each edge in the original graph is covered by exactly two copies of $P_{3}$, it can be shown $\omega \big( \mathcal{G}^{C_{3}}; S \big) \leq \omega^* \big( \mathcal{G}^{P_3} \big)$ for any strategy $S$.

\begin{figure}[ht]
    \centering
    \begin{tikzpicture}
        \coordinate (1) at (90:4em) {};
        \coordinate (2) at (210:4em) {};
        \coordinate (3) at (330:4em) {};
        
        \foreach \i/\j in {1/2, 2/3, 3/1} {
            \draw[whiteEdge] (\i.center) -- (\j.center);
            \draw[blackEdge] (\i.center) -- (\j.center);
        }

        \foreach \i/\n in {1/A, 2/B, 3/C} {
            \node[labeledWhiteVertex] at (\i.center) {};
            \node[labeledBlackVertex] at (\i.center) {\n};
        }

        \coordinate (p1) at ([xshift=-1em]1) {};
        \coordinate (p2) at (210:6em) {};
        \coordinate (p3) at ([yshift=-1em]3) {};

        \draw[-, rounded corners, edge, line cap=round, draw=figureColorPalette1] (p1) -- (p2) -- (p3);

        \coordinate (p1) at ([yshift=-1.5em]2) {};
        \coordinate (p2) at (330:7em) {};
        \coordinate (p3) at ([xshift=1.5em]1) {};

        \draw[-, rounded corners, edge, line cap=round, draw=figureColorPalette2] (p1) -- (p2) -- (p3);

        \coordinate (p1) at ([xshift=2.25em]3) {};
        \coordinate (p2) at (90:7.75em) {};
        \coordinate (p3) at ([xshift=-2.25em]2) {};

        \draw[-, rounded corners, edge, line cap=round, draw=figureColorPalette3] (p1) -- (p2) -- (p3);
    \end{tikzpicture}
    \caption{An example of a covering of the $3$-cycle graph $C_3$ by weighted copies of the path graph $P_3$. Each $P_3$ subgraph in $C_3$ (surrounding lines) has the weight $\nicefrac{1}{2}$. Each edge in $C_3$ is covered by exactly two $P_3$ subgraphs.}
    \label{fig:fractionalP3Decomposition}
\end{figure}

Hence, we see that if graph $H$ can be suitably decomposed using weighted copies of $P_3$, then we can readily conclude $\omega^*(\mathcal{G}^{H}) \leq \omega^*(\mathcal{G}^{P_{3}})$. 

\begin{lem}\label{lem:quantum_upperbound}
    If the quantum value of a game on $P_{3}$ is bounded by $\nu$, then the quantum value is bounded above by $\nu$ for all graphs not in $\mathcal{T}_{k}$. 
\end{lem}

Our proof, which is given in supplementary material, establishes a correspondence between $P_{3}$-decompositions of a graph and fractional perfect matchings of its line graph. We then use a well-known characterization of fractional perfect matchings to derive exactly when a graph has a $P_{3}$-decomposition. The full proof is contained in the supplementary material.

Combining \cref{lem:classical_graph_strategy} and \cref{lem:quantum_upperbound}, we obtain our main result.

\begin{thm}\label{thm:Main1}
     If there is no quantum advantage of a game $\mathcal{G}$ on $P_{3}$, then the game does not have quantum advantage on any graph $H$ such that $H \to \mathcal{S}_{\mathcal{G}}$ and $H$ is not in $\mathcal{T}_k$, for all $k \geq 2$.
\end{thm}

Next we show that we can further restrict our attention using the inductive nature of the families of graphs $\mathcal{T}_{k}$. As shown in \cref{fig:T3andT4}, understanding the behavior on $P_3$ and graphs from $\mathcal{T}_k$ for a fixed value of $k$, is sufficient to conclude the behavior on $\mathcal{T}_{k+1}$. The details are given in the supplementary material.

\begin{thm}\label{thm:Main2}
    Suppose a game $\mathcal{G}$ does not have quantum advantage on $P_{3}$ and all graphs in $\mathcal{T}_{k}$, for some $k$, then the game has no quantum advantage on all graphs $\mathcal{T}_{k+1}$.
\end{thm}

Noticing that $\mathcal{T}_{2} = \{P_{4}\}$, we obtain the following corollary which characterizes monogamous games, i.e games for which any extra connectivity is sufficient to completely inhibit the expression of nonlocality.

\begin{cor}\label{cor:Main3}
   A game has no advantage on $P_{3}$ and $P_{4}$ if and only if there is no advantage on any graph $H$ such that  $H \to \mathcal{S}_{\mathcal{G}}$.
\end{cor}

Next we will exhibit that numerous well-studied nonlocal games satisfy the above criteria and that the CHSH game does not. 

%% file: sections/examples.tex
\section{Illustrative examples}

\paragraph{The CHSH nonlocal game.}

The Bell operator for the CHSH game is given by,
\begin{equation*}
    \mathcal{B} = A_0B_0 + A_0B_1 + A_1B_0 - A_1B_1.
\end{equation*}
This game has classical value of $\nicefrac{3}{4}$ on all graphs, since the strategy graph of this game as a loop (see \Cref{fig:strategy_graph}), and the quantum value is given by the Tsirelson bound of $\tfrac{1}{2} + \tfrac{1}{2 \sqrt{2}}$ \cite{Cirelson1980}. 

Here we provide the sums-of-squares proof that the CHSH game has no quantum advantage when played on the graph $P_3$:
\begin{equation} \label{eq:CHSH_P3_monogamy}
    2 - \frac{1}{2} \big( \mathcal{B}_{AB} + \mathcal{B}_{BC} \big) = Q^2_1 + Q^2_2,
\end{equation}
where $Q_1$ and $Q_2$ are as follows
\begin{align*}
    Q_1 &\coloneqq \frac{1}{2\sqrt{2}} \brac{ A_0 (B_0 - B_1) + C_1 (B_0 + B_1) - 2 A_0 C_1 } \\
    Q_2 &\coloneqq \frac{1}{2\sqrt{2}} \left( A_1 (B_0 + B_1) +  C_0 (B_0 - B_1) - 2A_1 C_0 \right).
\end{align*}

We show below that there exists a quantum strategy which wins with probability $\tfrac{1}{2} +\tfrac{\sqrt{10}}{12} \approx 0.76 > \nicefrac{3}{4}$ on the $P_4$ graph. Furthermore, this strategy is optimal. We, once again, show this upper bound via a sum-of-squares decomposition,
\begin{equation} \label{eq:CHSH_P4_monogamy}
    \frac{2 \sqrt{10}}{3} - \frac{1}{3} \big(\mathcal{B}_{AB} + \mathcal{B}_{BC} + \mathcal{B}_{CD} \big) = R^2_1 + R^2_2 + R^2_3 + R^2_4,
\end{equation}
where $R_{1}, \ldots, R_{4}$ are given in the appendix, and $\tfrac{2 \sqrt{10}}{3}$ is the bias corresponding to the optimal winning probability. 

The quantum state $\rho$ and measurements which saturates this bound are given in the appendix. The positive-partial-transpose (PPT) criterion \cite{Peres1996,Horodecki1996} reveals that the pairs $(\rho_{AB}, \rho_{CD})$ and $(\rho_{AC}, \rho_{BD})$ are entangled while the $(\rho_{AD}, \rho_{BC})$ are not. The local values of the CHSH game on $P_4$ are given in \cref{fig:P3andP4}.

\begin{figure}[t]
    \centering
    \hfill
        \begin{tikzpicture}
                \coordinate (1) at (0,0) {};
                \coordinate (2) at (0,4em) {};
                \coordinate (3) at (4em,0) {};
                \coordinate (4) at (4em,4em) {};
                \coordinate (5) at (8em,0) {};
                \coordinate (6) at (8em,4em) {};

                \node at (4em,-2em) {$P_6$};
        
                \foreach \i/\j in {1/2, 3/4, 5/6} {
                    \draw[whiteEdge] (\i.center) -- (\j.center);
                    \draw[blackEdge] (\i.center) -- (\j.center);    
                }

                \foreach \i/\j in {2/4, 3/5} {
                    \draw[whiteEdge] (\i.center) -- (\j.center);
                    \draw[blackEdge] (\i.center) -- (\j.center);    
                }
        
            \foreach \i in {1,...,6}
                \node[blackVertex] at (\i.center) {};
        \end{tikzpicture}
    \hfill
        \begin{tikzpicture}
                \coordinate (1) at (0,0) {};
                \coordinate (2) at (0,4em) {};
                \coordinate (3) at (4em,0) {};
                \coordinate (4) at (4em,4em) {};
                \coordinate (5) at (8em,0) {};
                \coordinate (6) at (8em,4em) {};

                \node at (4em,-2em) {$\bigstar_{1,2,2}$};
        
                \foreach \i/\j in {1/2, 3/4, 5/6} {
                    \draw[whiteEdge] (\i.center) -- (\j.center);
                    \draw[blackEdge] (\i.center) -- (\j.center);    
                }

                \foreach \i/\j in {2/4, 4/6} {
                    \draw[whiteEdge] (\i.center) -- (\j.center);
                    \draw[blackEdge] (\i.center) -- (\j.center);    
                }
        
            \foreach \i in {1,...,6}
                \node[blackVertex] at (\i.center) {};
        \end{tikzpicture}
    \hfill{}
    
    \bigskip
    
        \begin{tikzpicture}
          \coordinate (1) at (0,0) {};
            \coordinate (2) at (0,4em) {};
            \coordinate (3) at (4em,0) {};
            \coordinate (4) at (4em,4em) {};
            \coordinate (5) at (8em,0) {};
            \coordinate (6) at (8em,4em) {};
            \coordinate (7) at (12em,0) {};
            \coordinate (8) at (12em,4em) {};

            \node at (6em,-2em) {Covering of an element in $\mathcal{T}_4$};

            \foreach \i/\j in {1/2, 1/3} {
                \draw[whiteEdge] (\i.center) -- (\j.center);
                \draw[dash=on 0.3pt off 4.6pt phase 1.1pt, line cap=round, edge, draw=figureColorPalette1] (\i.center) -- (\j.center);

            }

            \foreach \i/\j in {3/4, 4/6, 5/6, 6/8, 7/8} {
                \draw[whiteEdge] (\i.center) -- (\j.center);
                \draw[densely dashed, edge, draw=figureColorPalette2] (\i.center) -- (\j.center);
            }

            \foreach \i in {1,...,8}
                \node[blackVertex] at (\i.center) {};
        \end{tikzpicture}
    \caption{The graph family $\mathcal{T}_3$ contains two graphs: the path graph $P_6$ and the tree $\bigstar_{1,2,2}$. Below, a covering of graph in $\mathcal{T}_4$ using one $P_3$ subgraph, formed using the 3 nodes on the left, and a subgraph from $\mathcal{T}_{3}$ on the right.}
    \label{fig:T3andT4}
\end{figure}

As a consequence of the quantum advantage on $P_4$, \cref{cor:Main3} cannot be applied for the CHSH game. Instead, we will use the \cref{thm:Main2} with $k = 3$ to conclude the behavior on all other graphs, since for all graph $H$, we have $H \to \mathcal{S}_{\mathcal{G}}$. The graph family $\mathcal{T}_3$ contains only two graphs (up to isomorphism) as shown in \cref{fig:T3andT4}. 

Using the semidefinite hierarchy due to Navascués, Pironio, and Acín \cite{navascues2007bounding, navascues2008convergent}, we confirm numerically no quantum advantage for the CHSH game on either graph. Thus, by \cref{thm:Main2}, CHSH exhibits nonlocality only on the graphs $P_2$ and $P_4$.

We would like to note that the monogamy relation provided above of $\mathcal{B}_{AB} + \mathcal{B}_{BC} + \mathcal{B}_{CD} \leq 2\sqrt{10}$ is implied by the previously known quadratic Bell monogamy relation of $\mathcal{B}_{AB}^{2} + \mathcal{B}_{BC}^{2} \leq 8$ given by Toner and Verstraete \cite{0611001}. Indeed, by the Cauchy-Schwarz inequality, 
\begin{align*}
    \mathcal{B}_{AB} + \mathcal{B}_{BC} + \mathcal{B}_{CD} &\leq \sqrt{\mathcal{B}_{AB}^{2} + 2\mathcal{B}_{BC}^{2} + \mathcal{B}_{CD}^{2}} \sqrt{2 + \frac{1}{2}} \\
        &\leq \sqrt{16}\sqrt{2 + \frac{1}{2}} \\
        &= 2\sqrt{10}.
\end{align*}
As such, this monogamy relation does not imply any more information on the feasibility of certain quantum correlations more than what is implied by the quadratic Bell monogamy relations. In contrast to this, our numerical monogamy relations of $\mathcal{B}_{AB} + \mathcal{B}_{BC} + \mathcal{B}_{CD} + \mathcal{B}_{DE} + \mathcal{B}_{EF} \leq 10$ and $\mathcal{B}_{AB} + \mathcal{B}_{BC} + \mathcal{B}_{CD} + \mathcal{B}_{CF} + \mathcal{B}_{EF} \leq 10$ do actually provide \emph{new} Bell monogamy relations which cannot be derived from any previously known monogamy relations. Consider, for example, setting $\mathcal{B}_{AB}=\mathcal{B}_{CE}=\mathcal{B}_{EF}=\frac{61}{26}$, and $\mathcal{B}_{BC}=\mathcal{B}_{DE}=\frac{41}{26}$. This point will be in the feasible region implied by $\mathcal{B}^2_{XY} +\mathcal{B}^2_{YZ} \leq 8$ for any three consecutive $X,Y,Z$. On the other hand, $\mathcal{B}_{AB}+\mathcal{B}_{BC}+\mathcal{B}_{CD}+\mathcal{B}_{DE}+\mathcal{B}_{EF} > 10$ (see~\cref{fig:new_monogamy_relation}).

\begin{figure}[ht]
    \centering
    \begin{tikzpicture}
        \begin{axis}[
            xmin=0, xmax=3,
            ymin=0, ymax=3,
            xlabel={$\mathcal{B}_{AB} \;=\; \mathcal{B}_{CD} \;=\; \mathcal{B}_{EF}$},
            ylabel={$\mathcal{B}_{BC} \;=\; \mathcal{B}_{DE}$},
            xtick={0,1,2,3},
            ytick={0,1,2,3},
        ]
            \addplot [
                ultra thick, 
                draw=figureColorPalette2!70, 
                fill=figureColorPalette2!30,
            ] 
            coordinates {(0,0) (0,5) (10/3,0) (0,0)};
            
            \addplot [
                thick,
                domain=0:2*pi, 
                samples=200, 
                draw=figureColorPalette1!70, 
                fill=figureColorPalette1!30,
            ]
            ({sqrt(8)*cos(deg(x))}, {sqrt(8)*sin(deg(x))});
            
            \addplot [
                thick, 
                draw=figureColorPalette2!70, 
                fill=none,
            ] 
            coordinates {(0,0) (0,5) (10/3,0) (0,0)};
            
            \addplot[
              only marks, 
              mark=*,
              mark size=1.5pt,
              red!90,
            ] 
            coordinates {(61/26,41/26)};
        \end{axis}
    \end{tikzpicture}
    \caption{A slice of the correlation set for the CHSH game with six players on the path graph $P_6$ is shown, corresponding to the subspace defined by $AB = CD = EF$ and $BC = DE$. The disk of radius $2\sqrt{2}$ represents the intersection of the three cylinders imposed by the quadratic monogamy relations, while the linear region is defined by our new monogamy relation. Note that some points, although lying within the intersection of the cylinders, do not satisfy our monogamy relation, for instance the point $(\nicefrac{61}{26},\nicefrac{41}{26})$.}
    \label{fig:new_monogamy_relation}
\end{figure}

\paragraph{The Odd Cycle nonlocal game.}

We also remark that for the well-known $n$-vertex Odd Cycle game \cite{Cleve}, such average case advantage is not possible on any graph other than $P_2$. First note that the optimal classical strategy, achieving $1-\frac{1}{2n}$, is symmetric and thus by \cref{lem:classical_graph_strategy} the classical value is preserved on all graphs. Using the NPA hierarchy one can establish that this value upperbound the quantum value on both $P_3$ and $P_4$ graphs and thus by \cref{cor:Main3}, this bound holds for all graphs.

\subsection{Polygamous Nonlocality}\label{subsec:polygamous}

In this section, we resolve the question about the existence of a nonlocal game (with a quantum advantage) whose quantum winning probability remains equal when extending the standard two-player configuration to three players distributed along a path graph, i.e. $\omega^*(\mathcal{G}) = \omega^*(\mathcal{G}^{P_3})$.

We provide a positive solution to this question by introducing the following construction. Let $\mathcal{G}$ denote any two-player nonlocal game satisfying $\omega^*(\mathcal{G})=1$ and $\omega(\mathcal{G}) \leq \frac{1}{4}$. We then define a new nonlocal game, $\mathcal{G}^{\vee 2}$, as two simultaneous instances of $\mathcal{G}$. The referee's input is a pair of questions for each instance of $\mathcal{G}$, and the players' outputs are a corresponding pair of answers. The winning condition is satisfied if at least one of the two instances is won. Formally, the referee's predicate is  
\begin{align*}
    &v \big( (x_1, x_2), (y_1, y_2), (a_1, a_2), (b_1, b_2) \big) \\ &= v_{\mathcal{G}}(x_1, y_1, a_1, b_1) \vee v_{\mathcal{G}}(x_2, y_2, a_2, b_2), 
\end{align*}
where $v_{\mathcal{G}}(\cdot)$ is the referee predicate for the original $\mathcal{G}$ game. 

The quantum winning probability for $\mathcal{G}$ with two players is $1$, and In the case of three players arranged as $P_3$ the quantum value remains $1$. To see this let $\ket{\mathcal{G}}$ be a bipartite state that wins the game $\mathcal{G}$ with probability $1$. The global state shared by the three players $A,B,C$ is given by
\begin{equation*}
    \underbrace{\big( \ket{\mathcal{G}}_{AB} \otimes \ket{0}_C \big)}_{= \ket{\psi}} \otimes \underbrace{\big( \ket{0}_A \otimes \ket{\mathcal{G}}_{BC} \big)}_{= \ket{\phi}}.
\end{equation*}
The measurement strategy is as follows: for the first instance, the players use the state $\ket{\psi}$, thus $A$ and $B$ use $\ket{\mathcal{G}}_{AB}$ and win with probability 1; for the second instance, the players use the state $\ket{\phi}$, thus $B$ and $C$ use $\ket{\mathcal{G}}_{BC}$, also winning with probability 1. Consequently, on each edge of $P_3$ at least one instance of $\mathcal{G}$ is always won with certainty, giving $\omega^*(\mathcal{G}^{P_3}) = 1$.

To see that this is not possible for any classical strategy note that for two players
\begin{align*}
   \omega \big( \mathcal{G}^{\vee 2} \big) &= \Pr( \text{W}_1 \wedge  \text{W}_2) + \Pr(\text{W}_1 \wedge  \text{L}_2) + \Pr(\text{W}_2 \wedge  \text{L}_1) \\ &\leq 2\Pr(\text{W}_1) + \Pr(\text{W}_2) \leq \frac{3}{4},
\end{align*}
where $\text{W}_i$ and $\text{L}_i$ are the events that the $i$-th instance of the game $\mathcal{G}$ is won or lost, respectively. 
Since $P_3$ is a bipartite graph, there is a graph homomorphism from $P_3$ to the strategy graph for $\mathcal{G}$. Using \cref{lem:classical_graph_strategy} it follows that $\omega(\mathcal{G}^{P_3}) = \omega(\mathcal{G}) \leq \frac{3}{4}$. 

To show that such a game $\mathcal{G}$ exists, let MS denote the symmetric magic square game \cite{Mermin, Peres, SynBCS}, which has a classical winning probability $\omega(\text{MS}) = \tfrac{35}{36}$ and an quantum winning probability $\omega^*(\text{MS}) = 1$. Then take $\mathcal{G} = \text{MS}^{n}$, to denote the $n$-th fold parallel repetition of MS. By Raz's parallel repetition theorem \cite{RazParallel} we have that $\omega(\mathcal{G}) \leq \frac{1}{4}$ for large enough $n$.

%% file: sections/discussion.tex
\section{Discussion}

In conclusion, in this letter, we provide a general graph-theoretic method which allows us to extend Bell monogamy relations on certain smaller configurations to an infinite family of configurations. We apply this result to the well-known games, CHSH and the Odd Cycle games. For CHSH, we find new Bell monogamy relations between 6 parties which cannot directly be inferred from known monogamy relations on $P_3$, and demonstrate that this relation is sometimes more informative. Finally we give the first example a game where the quantum value on $P_3$ is $1$ for both pairs of players simultaneously. 

We conclude this letter with future directions:
\begin{itemize}
    \item The quadratic Bell monogamy relation of Toner and Verstraete has the property that any numerical values that satisfy the inequality are physically realizable \cite{0611001}. Is there an analogous Bell monogamy relation for 6 parties?
    \item We can achieve quantum advantage on any degree-bounded graph. Is it possible to have a game which has advantage on \emph{all} graphs?
    \item One can consider extending instead $k$-player nonlocal games to $k$-uniform hypergraphs as considered in Tran et al. \cite{ramanathan2018bell}. Can we extend our graph decomposition techniques to this more general setting?
\end{itemize}

%% file: sections/acknowledgements.tex
\section{Acknowledgements}

\begin{acknowledgments}
    This research was initiated under the wonderful supervision of David Gosset and William Slofstra at the Institute for Quantum Computing during the summer URA program. D.C. thanks them for the many fruitful discussions and their mentorship. The authors thank Elie Wolfe, Daniel Díaz, and Maria Alanon for many helpful conversations and for their remarks on non-signaling correlations. We thank and Lewis Wooltorton for their comments on an earlier version of the paper. The authors also thank Aaron Tiskuisis for suggesting the use of the strategy graph when analysis the classical value of a game, which was included in the updated version of this work. Lastly, the authors thank Ravishankar Ramanathan for feedback on the construction of a polygamous game.  D.C. acknowledges the support of the Natural Sciences and Engineering Research Council of Canada through grant number RGPIN-2019-04198. A.M. is supported by NSERC Alliance Consortia Quantum grants, reference number: ALLRP 578455 - 22. D.R. acknowledges the support of the Air Force Office of Scientific Research under award number FA9550-20-1-0375.
\end{acknowledgments}

%% file: sections/supplementary.tex
\section{Supplementary Material}  

\subsection{CHSH on \texorpdfstring{$P_4$}{P4}}

First, we give the sum-of-squares decomposition for $\mathrm{CHSH}^{P_{4}}$:
\begin{equation*}
    2 \sqrt{10} - \big(\mathcal{B}_{AB} + \mathcal{B}_{BC} + \mathcal{B}_{CD} \big) = R^{2}_{1} + R^2_2  + R^2_3  + R^2_4,
\end{equation*}
with $R_1, \ldots, R_4$ defined by
\begin{widetext}
    \begin{align*}
        R_1 &\coloneqq C_0 (B_0 - B_1) + C_1 (B_0 + B_1) - 2 D_0 (C_0 + C_1) + 2 D_1 (C_0 - C_1) \\
        R_2 &\coloneqq 1 - \frac{1}{2 \sqrt{10}} \big( (B_0 (A_0 + A_1) + B_1 (A_0 - A_1) + C_0 (B_0 +  B_1) + C_1 (B_0 - B_1) + D_0(C_0 + C_1) + D_1 (C_0 - C_1) \big) \\
        R_3 &\coloneqq \frac{1}{3} \big( B_0 (3 A_0 -  A_1) + B_1 (3 A_0 + A_1) - C_0 (B_0 + B_1) + C_1 (B_0 - B_1) - D_0 (C_0 + C_1) - D_1 (C_0 - C_1) \big) \\
        R_4 &\coloneqq \frac{1}{4} \big( 4 A_1 (B_0 - B_1) + C_0 (B_0 + B_1) - C_1 (B_0 - B_1) - 2 D_0 (C_0 + C_1) - 2 D_1 (C_0 - C_1) \big).
    \end{align*}
\end{widetext}

The optimal value $2 \sqrt{10}$ for CHSH over $P_4$ can be obtained for $\mathcal{B}_{AB}$ and $\mathcal{B}_{CD}$ both equal to the value $\frac{4\sqrt{2}}{\sqrt{5}}$ and $\mathcal{B}_{BC}$ equals to the value $\frac{2\sqrt{2}}{\sqrt{5}}$, using the following observables
\begin{align*}
    A_0 = C_0 &= \sigma_z, & A_1 = C_1 &= \sigma_z \\
    B_0 = D_0 &= \frac{\sigma_x + \sigma_z}{\sqrt{2}}, & B_1 = D_1 &= \frac{\sigma_x - \sigma_z}{\sqrt{2}}.
\end{align*}
where $\sigma_i$ are the Pauli matrices.
We denotes by $(\textsc{ab})$ the swap operator which interchanges the system between $A$ and $B$. Then the quantum state $\rho$ which saturates this bound is
\begin{align*}
    \rho \coloneqq& S \big( \ket{\Omega}\!\bra{\Omega} \otimes \ket{\Omega}\!\bra{\Omega} \big) S^* \\
    =& \frac{1}{20} \sum_{i=0}^{1} (3\sqrt{5} + 5) \ket{iiii} + ( 5 + \sqrt{5} ) \ket{ii\overline{ii}} \\ &\quad\qquad + (5 - \sqrt{5}) \ket{i\overline{ii}i} + (3\sqrt{5}-5)\ket{i\overline{i}i\overline{i}} 
\end{align*}
where $S$ is the following Hermitian (non-unitary) sum of swap operators:
\begin{align*}
    S \coloneqq \frac{1}{20} \Big( &- (5 + \sqrt{5}) (\textsc{ad})(\textsc{bc}) \\
    &+ (5 - 3 \sqrt{5}) (\textsc{ad})  + (-5 + \sqrt{5}) (\textsc{bc}) \Big).
\end{align*}

\subsection{Proofs of Lemmas and Theorems}

Before proceeding with the proofs, let us first give the definitions of fractional $P_3$-decomposition, fractional perfect matching, and line graph.

\begin{defn}
    Let $P$ be the collection of $P_{3}$ subgraphs in $H$. A graph $H$ has a \emph{fractional $P_{3}$-decomposition} if there exists a function $f : P \to [0,1]$ such that for each edge $e \in H$, $\sum_{\substack{K \in P \\ e \in K}} f(K) = 1$, where the sum is over all $P_{3}$ subgraphs $K$ in $P$ containing the edge $e$.
\end{defn}

\begin{defn}
    A graph $H=(V,E)$ has a \emph{fractional perfect matching} if there exists a function $f : E \to [0,1]$ such that for each vertex $v \in V$, $\sum_{(u,v) \in E} f \big( (u,v) \big) = 1$.
\end{defn}

\begin{defn}
    Given a graph $H$, the \emph{line graph of $H$} is the graph $L(H)$ whose vertices are edges of $H$ and two vertices of $L(H)$ are incident if the corresponding edges are incident in $H$.
\end{defn}

The following lemma characterizes exactly which graph $H$ has a fractional $P_{3}$-decomposition.

\begin{lem}\label{lem:graph_decomp}
     A connected graph $H$ has a fractional $P_{3}$-decomposition if and only if $H$ does not belong to $\mathcal{T}_{k}$ for any $k$.
\end{lem}
\begin{proof}
    By definition of the line graph, there is a objective correspondence between the edges of $H$ and the vertices of $L(H)$. Hence, we have a objective correspondence between $P_{3}$ subgraphs of $H$ and edges of $L(H)$. This allows us to map between fractional $P_{3}$-decompositions on $H$ and fractional perfect matchings of $L(H)$, thereby concluding that $H$ has a fractional $P_{3}$-decomposition if and only if $L(H)$ has a fractional perfect matching.

    Now, given a graph $H$, denote $i(H)$ the number of isolated vertices (i.e. degree zero vertices) of $H$ and $\tilde{i}(H)$ the number of isolated $P_{2}$ subgraphs of $H$. It is a standard result \cite{Scheinerman1997FractionalGT} from graph theory that a graph $H$ has a fractional perfect matching if and only if $i(H \setminus S) \leq \abs{S}$ for all $S \subset V(H)$. Using this and the bijective correspondence between edges of $H$ and vertices of $L(H)$, we have that $H$ has a fractional $P_{3}$-decomposition if and only if $\tilde{i}(H \setminus F) \leq \abs{F}$ for all $ F \subset E(H)$.
    
    Finally, to prove the lemma, assume that the graph $H$ does not have a fractional $P_{3}$ decomposition. Then there exists an edge set $F$ such that $\abs{F} < \tilde{i}(H \setminus F)$. So $H \setminus F$ has $\tilde{i}(H \setminus F)$ disjoint copies of $P_{2}$ and some other connected components $H_{1}, \dots, H_{m}$ which are not $P_{2}$. Since we began with a connected graph $H$, we have $\abs{F} \geq m + \tilde{i}(H \setminus F) - 1$. Therefore, $m = 0$ and so $H \in \mathcal{T}_{\tilde{i}(H\setminus F)}$. Conversely, if $H \in \mathcal{T}_{k}$ for some $k$, then let $F$ be the $k-1$ edges joining the $k$ copies of $P_{2}$ in the construction of $H$. Thus, clearly $\abs{F} < \tilde{i}(H \setminus F)$, and so $H$ does not have a fractional $P_{3}$-decomposition.
\end{proof}

{
\color{black}
\begin{proof}[Proof of \Cref{lem:quantum_upperbound}]
     Let $\mathcal{G}$ be a game, $H=(V,E)$ a graph not belonging to any $\mathcal{T}_{k}$, and $\nu$ such that $\omega^* \big( \mathcal{G}^{P_{3}} \big) \leq \nu$. Therefore, by \Cref{lem:graph_decomp} there exists a fractional $P_{3}$-decomposition $f$ of $H$. Then for any quantum strategy $S$,
    \begin{align*}
        \omega \big( \mathcal{G}^{H}; S \big) &= \frac{1}{\abs{E}} \sum_{ e \in E } \omega \big( \mathcal{G}; \restr{S}{e} \big) \\
                  &= \frac{1}{\abs{E}} \sum_{ e \in E } \sum_{\substack{K \in P \\ e \in K}} f(K) \; \omega \big( \mathcal{G}; \restr{S}{e} \big) \\
                  &= \frac{1}{\abs{E}} \sum_{K \in P} f(K) \sum_{ e \in K } \omega \big( \mathcal{G}; \restr{S}{e} \big) \\
                  &\leq \frac{1}{\abs{E}} \sum_{K \in P} f(K) \; \nu \\
                  &= \nu,
    \end{align*}
    where $\restr{S}{e}$ denotes the restriction of the quantum strategy $S$ to players on the edge $e$.
\end{proof}

\begin{proof}[Proof of \Cref{thm:Main1}]
    Let $\mathcal{G}$ be a game such that $\omega^* \big( \mathcal{G}^{P_{3}} \big) = \omega \big( \mathcal{G}^{P_{3}} \big)$, and let $H=(V,E)$ a graph such that $H \to \mathcal{S}_{\mathcal{G}}$ and $H$ does not belong to any $\mathcal{T}_{k}$. Then by \Cref{lem:quantum_upperbound}, for any quantum strategy $S$, we have $\omega \big( \mathcal{G}^{H}; S \big) \leq \omega \big( \mathcal{G}^{P_{3}} \big)$. But by \Cref{lem:classical_graph_strategy}, $\omega \big( \mathcal{G}^{P_{3}} \big) = \omega \big( \mathcal{G}^{H} \big)$.
\end{proof}
}

\Cref{lem:graph_decomp} tells us that for any graph outside of $\bigcup_{k} \mathcal{T}_{k}$, the game $\mathcal{G}$ has no quantum advantage. Furthermore, if we are able to show that, for a fixed $k$, and for all graphs in $\mathcal{T}_{k}$, the game $\mathcal{G}$ has no quantum advantage, then $\mathcal{G}$ has no quantum advantage on larger graphs.

\begin{proof}[Proof of \Cref{thm:Main2}]
    \textcolor{black}{All the graphs in $\mathcal{T}_{k}$ are bipartite graphs, thus for all graph $H \in \mathcal{T}_{k}$ and all games $\mathcal{G}$, we have $H \to \mathcal{S}_{\mathcal{G}}$.}
    If the game $\mathcal{G}$ has no quantum advantage on $P_{3}$, then by \Cref{thm:Main1}, it suffices to consider the graphs in $\bigcup_{k} \mathcal{T}_{k}$. Suppose for some fixed $k$, the game $\mathcal{G}$ has no advantage on all graphs in $\mathcal{T}_{k}$. First notice that for all $j > k$, and any graph $H = (V, E) \in \mathcal{T}_{j}$, there is a graph $H_{1} = (V_{1}, E_{1}) \in \mathcal{T}_{k}$ such that $H_{1}$ is a subgraph of $H$. Let $H_{2} = (V_{2}, E \setminus E_{1})$, where $V_{2}$ is the set of vertices in $E \setminus E_{1}$. The graph $H_{2}$ is clearly also a subgraph of $H$. Furthermore, since $\abs{E}$ and $\abs{E_{1}}$ are odd (as they belong to $\bigcup_{k} \mathcal{T}_{k}$), then $\abs{E} - \abs{E_{1}}$ is even and hence $H_{2}$ has a fractional $P_{3}$ decomposition. Furthermore,
    \begin{align*} 
        \omega \big( \mathcal{G}^{H}; S \big) &= \frac{1}{\abs{E}} \sum_{e \in E} \omega \big( \mathcal{G}; \restr{S}{e} \big) \\
                                   &= \frac{\abs{E_{1}}}{\abs{E}} \omega \big( \mathcal{G}^{H_{1}}; \restr{S}{H_{1}} \big) + \frac{\abs{E}-\abs{E_{1}}}{\abs{E}} \omega \big( \mathcal{G}^{H_{2}};\restr{S}{H_{2}} \big),
    \end{align*}
    where $\restr{S}{H_1}$ and $\restr{S}{H_2}$ denote the restriction of the quantum strategy $S$ to players on the subgraphs $H_1$ and $H_2$ respectively. Hence there is no quantum advantage on $H$.
\end{proof}